\documentclass[12pt]{amsart}

\usepackage{amsmath, amsthm, latexsym, amssymb, amsfonts, epsf, xcolor}
\usepackage[hidelinks]{hyperref}
\usepackage{cite}
\usepackage{caption, subcaption, enumerate, color}
\usepackage{multirow,multicol}
\usepackage{mathtools}
\usepackage{bm}
\usepackage{blkarray}
\usepackage{booktabs}
\usepackage{framed}
\usepackage{setspace}
\usepackage{tikz}
\usepackage{graphicx}
\usepackage{wrapfig}
\usepackage{float}
\usepackage[dvips]{epsfig}
\usepackage{bbm}
\usepackage{url}
\usepackage{booktabs}
\usepackage{bibentry}
\usepackage[normalem]{ulem}
\usepackage{comment}
\usepackage{lingmacros}
\usepackage{cleveref}
\usepackage{extarrows}
\usepackage{multirow}
\usepackage{mathtools}
\usepackage{enumerate}
\usepackage{array}
\usepackage{bigstrut}
\usepackage{exscale,relsize}
\usepackage{graphicx}
\usepackage[utf8]{inputenc}
\usepackage[linesnumbered, ruled]{algorithm2e}

\DeclarePairedDelimiter\abs{\lvert}{\rvert}%

\newcommand{\ev}{\operatorname{ev}}

\newcommand{\F}{{\mathbb F}}

\DeclareMathOperator{\ini}{in}

\numberwithin{equation}{section}
\newtheorem{theorem}{Theorem}[section]
\newtheorem{lemma}[theorem]{Lemma}

\newtheorem{corollary}[theorem]{Corollary}

\theoremstyle{definition}
\newtheorem{definition}[theorem]{Definition} 
\newtheorem{remark}[theorem]{Remark}
\newtheorem{example}[theorem]{Example}

\newcommand{\rmv}[1]{}

\DeclareMathOperator{\supp}{supp}

\newcommand{\PP}{{\mathbb{P}}}
\newcommand{\A}{{\mathbb{A}}}
\newcommand{\X}{{\mathbb{X}}}
\newcommand{\cda}[1]{C_d^{\A^{#1}_0,S}(#1)}

\newcommand{\cdx}{C_d^{\mathbb{X},S}}
\newcommand{\cdxx}{C_{\leq d}^{\mathbb{X},S}}
\newcommand{\cdp}{C_d^{\PP,S}}
\newcommand{\fq}{\mathbb{F}_q}

\begin{document}


\title[Cartesian square-free codes]{Cartesian square-free codes}

\author[C. Carvalho]{C\'{\i}cero Carvalho}
\address[Cicero Carvalho]{Department of Mathematics\\ University of Uberlandia\\ Uberlandia,  Brazil}
\email{cicero@ufu.br}

\author[H. L\'opez]{Hiram H. L\'opez}
\address[Hiram H. L\'opez]{Department of Mathematics\\ Virginia Tech\\ Blacksburg, VA USA}
\email{hhlopez@vt.edu}

\author[R. San-Jos\'e]{Rodrigo San-Jos\'e}
\address[Rodrigo San-Jos\'e]{Department of Mathematics\\ Virginia Tech\\ Blacksburg, VA USA. \textit{Previous address:} IMUVA-Mathematics Research Institute, Universidad de Valladolid, 47011 Valladolid (Spain)}
\email{rsanjose@vt.edu}

\thanks{C\'{\i}cero Carvalho was partially supported by Fapemig - grant APQ-01430-24, and CNPq grant
PQ 308708/2023-7. Hiram H. L\'opez was partially supported by the NSF Grants DMS-2401558 and DMS-2502705.
Rodrigo San-Jos\'e was partially supported by Grant PID2022-138906NB-C21 funded by MICIU/AEI/10.13039/501100011033 and by ERDF/EU, and FPU20/01311 funded by the Spanish Ministry of Universities.}

\dedicatory{Festschrift on the occasion of the 80th anniversary of the Mexican Mathematical Society}

\begin{abstract}
The generalized Hamming weights (GHWs) of a linear code $C$ extend the concept of minimum distance, which is the minimum cardinality of the support of all one-dimensional subspaces of $C$, to the minimum cardinality of the support of all $r$-dimensional subspaces of the code. In this work, we introduce Cartesian square-free codes, which are linear codes generated by evaluating square-free monomials over a Cartesian set. We use commutative algebraic tools, specifically the footprint bound, to provide explicit formulas for some of the GHWs of this family of codes, and we show how we can translate these results to evaluation codes over the projective space.
\end{abstract}

\maketitle

\section{Introduction} \label{S:intro}
An $[n,k,d]$ linear code $C$ over a finite field with $q$ elements $\F_q$ is a $k$-dimensional subspace of $\F_q^n$ with minimum distance
$d:=\min \{ d({\bm c},{\bm c}') \ : \ {\bm c}, {\bm c}' \in C, {\bm c} \neq {\bm c}' \},$ where 
$d({\bm c},{\bm c}'):=\min \left| \{ i : c_i \neq c_i' \} \right|$ denotes the Hamming distance. Linear codes, extensively studied for more than seventy-five years since the seminal paper by Shannon in 1948~\cite{MR0026286}, were originally introduced to reliably and efficiently communicate information from a source to a receiver through a noisy channel. Throughout the years, linear codes have found a series of applications such as satellite communication, cell phones, distributed computation, and cryptography, among many others. In 1991, Wei introduced generalized Hamming weights (GHWs) for linear codes to measure the performance of a type-II wiretap channel~\cite{weiGHW}. The GHWs extend the minimum distance of a code, which is the minimum cardinality of the support of all one-dimensional subspaces of $C$, to the minimum cardinality of the support of all $r$-dimensional subspaces of $C$. Specifically, the $r^{th}$-generalized Hamming weight of a code $C \subseteq \F_q^n$ is given by
$$
d_r(C):=\min \left\{ \abs{ \supp(C') }: C' \textnormal{ is a subcode of } C \textnormal{ of dimension } r \right\},
$$
where $1\leq r \leq k$ and
$
\supp(C):=\left\{ i: \exists c \in C \textnormal{ with } c_i \neq 0 \right\}=\bigcup_{c\in C}\supp(c).
$ The weight hierarchy of $C$ is defined
to be the set of integers $\{ d_r(C): 1 \leq r \leq k\}.$ In the same work~\cite{weiGHW}, Wei computed the weight hierarchy for several families of linear codes, including Reed–Muller, Hamming, and extended Hamming, simplex (dual of a Hamming), Reed–Solomon, and all maximum-distance-separable (MDS) codes.

Since 1991, the weight hierarchy has been extensively studied in the literature. Helleseth, Kl{\o}ve, and Ytrehus focused on the GHMs of binary linear codes in~\cite{helleseth1992ghw}. They gave bounds for these weights, provided a generalized Griesmer bound, determined the weight hierarchies for $k \leq 3$, and the optimal values of $d_r$ for $r \geq k-3$. In 1992, Feng, Tzeng, and Wei provided the GHWs for certain families of cyclic codes in~\cite{feng1992cyclic}. Janwa and Lal also studied the GHWs of cyclic codes in~\cite{janwa1997cyclic}. Other examples of families of codes for which the GHWs were studied are product codes~\cite{wei1993product}, trace codes~\cite{312185}, Hermitian codes~\cite{heijnen1998rm, barbero2000hermitian}, and more generally, decreasing norm-trace codes ~\cite{Camps-Moreno2025-xo}, $q$-ary Reed--Muller codes~\cite{heijnen1998rm}, projective Reed-Muller codes~\cite{sanjoseRecursivePRM,beelenGHWPRM,boguslavsky_GHWs_PRM} and their variants~\cite{villarreal_GHWs_certain_RM,villarreal_GHWs_graphs}, {M}elas codes and dual {M}elas codes~\cite{geer1994melas}, BCH codes~\cite{geer1994bch}, convolutional codes~\cite{567750, 10102487} and matrix-product codes \cite{sanjoseGHWMPC}.

Evaluation codes form an important family of linear codes defined by evaluating functions on certain points. Some examples are Reed-Solomon codes, which are given by evaluating polynomials in one variable up to a certain degree over a subset of points $A$ of $\mathbb{F}_q$. When we move from one to several variables, and from $A$ to $\mathbb{F}_q^m$, we obtain the Reed-Muller codes. In this work, we focus on the GHWs of the evaluation codes generated by the evaluation of square-free monomials over a Cartesian set. These are called Cartesian square-free codes, which include as particular cases all the previous families of square-free codes, such as binary Reed-Muller codes, or square-free codes over the torus $(\mathbb{F}_q^*)^m$ \cite{GHWsToricSquarefree,jaramillo}.

The content of this paper is as follows. In Section~\ref{S:Preli}, we introduce notation and elementary results about GHWs. In Section~\ref{footprint bout}, we describe the footprint bound, an essential tool from commutative algebra that we use to compute the GHWs of Cartesian square-free codes. In Section~\ref{formulas}, we use the results from Section~\ref{footprint bout} to give explicit formulas that provide the weight hierarchies of Cartesian square-free codes. In Section~\ref{projective}, we show that some of the results about the GHWs can be translated to evaluation codes over the projective space. 

\section{Preliminaries}
\label{S:Preli}
Let $A_1,\dots, A_m$ be a collection of non-empty subsets of $\fq$, the finite field with $q$ elements. We denote their Cartesian product by $\X=A_1\times \cdots \times A_m$. Let $n_i=\abs{A_i}$, for $1\leq i \leq m$. Without loss of generality, in what follows, we assume that $2\leq n_1\leq n_2\leq \cdots \leq n_m$. We denote $B_\X=\{0,\dots,n_1-1\}\times \cdots \times \{0,\dots,n_m-1\}$. We consider $S(m)=\{0,1\} ^m$, and we denote by $S_d(m)$ (resp., $S_{\leq d}(m)$) the tuples $(\alpha_1,\dots,\alpha_m)\in S$ such that $\sum_{i=1}^m\alpha_i=d$ (resp., $\sum_{i=1}^m\alpha_i\leq d$). We may omit the reference to $m$ if it is obvious from context, and denote these sets by $S_d$ and $S_{\leq d}$. 
We consider the codes
$$
\cdx:=\langle \ev_\X(x^\alpha): \alpha \in S_d \rangle, \;  \cdxx:=\langle \ev_\X(x^\alpha): \alpha \in S_{\leq d} \rangle.
$$

Note that $\cdx$ is the evaluation code generated by evaluating the square-free monomials of total degree $d$ in $m$ variables over the Cartesian set $\X$. Since we only evaluate square-free monomials, these are called Cartesian square-free codes. Their length is $n=\prod_{i=1}^m n_i$. For the codes $\cdx$, we have that $\dim \cdx=\binom{m}{d}$, the number of square-free monomials of degree $d$. By adding over all possible degrees $0\leq d' \leq d$, we also obtain $\dim \cdxx=\sum_{i=0}^d \binom{m}{i}$. Note that, for $d>m$, we have $S_d=\emptyset$ and $S_{\leq d}=V_{\leq m}^S$. Thus, we only consider these codes for $1\leq d \leq m$.

We are mainly interested in the generalized Hamming weights of these Cartesian square-free codes. Given a tuple or a vector $\alpha=(\alpha_1,\dots,\alpha_n)$, define 
$$
\supp(\alpha)=\{i: \alpha_i\neq 0\}.
$$
Similarly, the support of a  code $C\subset \fq^n$ is 
$$
\supp(C):=\left\{ i: \exists c \in C \textnormal{ with } c_i \neq 0 \right\}=\bigcup_{c\in C}\supp(c).
$$

\begin{definition} \label{def:ghw}
The $r^{th}$-generalized Hamming weight of an $[n,k,d]$ code $C$ is 
$$
d_r(C):=\min \left\{ \abs{ \supp(C') }: C' \textnormal{ is a subcode of } C \textnormal{ of dimension } r \right\},
$$
where $1\leq r \leq k$.
\end{definition}
The first generalized Hamming weight $d_1(C)$ is just the minimum distance of $C$. Wei obtained the following general properties of the generalized Hamming weights in \cite{weiGHW}.

\begin{theorem}[Monotonicity]\label{t:monotonia}
For an $[n,k]$ linear code $C$ with $k>0$, we have
$$
1\leq d_1(C)<d_2(C)<\cdots <d_k(C)\leq n.
$$
\end{theorem}

\begin{definition}
The dual of a code $C\subset \F_q^n$ is defined by
\[C^{\perp} : = \{ \bm{w}\in \F_q^n \ : \ \bm{w}\cdot\bm{c}=0 \text{ for all } \bm{c}\in C \}, \]
where $\bm{w}\cdot\bm{c}$ represents the Euclidean inner product.
\end{definition}
The following theorem relates the GHWs of a code and its dual.
\begin{theorem}[Duality]\label{ghwdual}
Let $C$ be an $[n,k]$ code. Then
$$
\{d_r(C):1\leq r\leq k\}=\{1,2,\dots,n\}\setminus \{n+1-d_r(C^\perp):1\leq r\leq n-k\}.
$$
\end{theorem}

\section{Footprint bound}
\label{footprint bout}
In this section, we introduce the footprint bound, which is the main tool we use to study the GHWs of square-free codes. For the tools we use from commutative algebra, we recommend the references \cite{cox,villarreal_book_monomial_algebras}. Fix a monomial order $\prec$ such that $x_1\prec x_2\prec \cdots \prec x_m$, and let $\mathbb{M}^m$ be the set of monomials of $\fq[x_1,\dots,x_m]$. For an ideal $I\subset \fq[x_1,\dots,x_m]$, we denote
$$
\Delta(I):=\left\{ M \in \mathbb{M}^m: M \not \in \ini(I) \right\}=\left\{ M \in \mathbb{M}^m: \ini(g) \nmid M \ \ \forall  g \in \mathcal G \right\},
$$ 
where $\mathcal G$ is a Grobner basis for $I$. The footprint bound states that 
$$
\abs{V(I)} \leq \abs{\Delta(I)}=\abs{\Delta(\ini(I))},
$$
with equality if $I$ is radical (see \cite[Thm. 6 and Prop. 7, Chapter 5 \textsection 3]{cox}). 
We denote by $\mathbb{M}^m_d$ (resp., $\mathbb{M}^m_{\leq d}$) the set $\{x^\alpha\in \mathbb{M}^m:\abs{\alpha}=d \} $ (resp., $\{x^\alpha\in \mathbb{M}^m:\abs{\alpha}\leq d\} $). Similarly, we denote $\mathbb{M}^S_d$ (resp., $\mathbb{M}^S_{\leq d}$) the set $\{x^\alpha\in \mathbb{M}^m:\alpha \in S_d\} $ (resp., $\{x^\alpha\in \mathbb{M}^m:\alpha \in S_{\leq d} $).
Let $\mathcal{N}_{r,d}:=\{ N=\{x^{\alpha_1},\dots,x^{\alpha_r}\}\subset\mathbb{M}_d^S: \alpha_i\neq \alpha_j \text{ for } i\neq j \}$ and $\mathcal{N}_{r,\leq d}:=\{ N=\{x^{\alpha_1},\dots,x^{\alpha_r}\}\subset\mathbb{M}_{\leq d}^S:\alpha_i\neq \alpha_j \text{ for } i\neq j \}$. In terms of codes, the footprint bound gives the following: 
\begin{equation}\label{eq:fpbound}
\begin{aligned}
d_r(\cdx)\geq n_1\cdots n_m-\max\left\{\abs{\Delta_\X(N)}, N \in \mathcal{N}_{r,d}\right\}=\min \{ \abs{\nabla_{\X}(N)}, N \in \mathcal{N}_{r,d}\}, \\
d_r(\cdxx)\geq n_1\cdots n_m-\max\left\{\abs{\Delta_\X(N)}, N \in \mathcal{N}_{r,\leq d}\right\}=\min \{ \abs{\nabla_{\X}(N)}, N \in \mathcal{N}_{r,\leq d}\},
\end{aligned}
\end{equation}
where $\Delta_\X(N):=\Delta(\ini(I(\X))+(x^\alpha\in N))$ and $\nabla_\X(N):=\Delta(I(\X))\setminus \Delta_\X(N)$. The set $\nabla_\X(N)$ is called the shadow of $N$. We show now that this bound is sharp for the GHWs of the codes $\cdx$ and $\cdxx$. First, we recall that the footprint bound is sharp for monomial decreasing Cartesian codes (in the sense of \cite{camps_polar_decreasing_cartesian,tillich_polar_decreasing}). This follows from the proof of \cite[Thm. 5.1]{eduardoGHWHyperbolic}, which works for any decreasing Cartesian code. This leads to the following result.

\begin{theorem}\label{t:fpsharp_leq_d}
Let $\X=A_1\times \cdots \times A_m$ be a Cartesian set, let $0\leq d \leq m$, let $1\leq r \leq \binom{m}{d}$, and let $M=\{x^{\alpha_1},\dots,x^{\alpha_r}\}\subset \mathbb{M}_{\leq d}^S$. Then there exists a set $F=\{f_1,\dots,f_r\}\subset \mathcal{L}(\mathbb{M}_{\leq d}^S)$ such that $\ini(f_i)=x^{\alpha_i}$, $1\leq i \leq r$, and 
$$
\abs{V_{\X}(F)}=\abs{\Delta_{\X}(M)}.
$$
As a consequence,
$$
d_r(\cdxx)=n_1\cdots n_m-\max\left\{\abs{\Delta_\X(N)}, N \in \mathcal{N}_{r,\leq d}\right\}=\min \{ \abs{\nabla_{\X}(N)}, N \in \mathcal{N}_{r,\leq d}\}.
$$
\end{theorem}

The previous theorem shows that we can compute the footprint bound instead of directly computing the GHWs and obtain the same result. This gives a significant computational advantage and transforms the problem of computing the GHWs into a combinatorial problem in which we have to find a set of $r$ monomials that minimizes the cardinality of the shadow. We can obtain similar results for the homogeneous case under additional hypotheses.

\begin{theorem}\label{t:fpsharp_homogeneous_d}
Let $\X=A_1\times \cdots \times A_m$ be a Cartesian set such that $0\in A_i$, $1\leq i \leq m$. Let $1\leq d \leq m$, let $1\leq r \leq \binom{m}{d}$, and let $M=\{x^{\alpha_1},\dots,x^{\alpha_r}\}\subset \mathbb{M}_d^S$. Then there exists a set $F=\{f_1,\dots,f_r\}\subset \mathcal{L}(\mathbb{M}_d^S)$ such that $\ini(f_i)=x^{\alpha_i}$, $1\leq i \leq r$, and 
$$
\abs{V_{\X}(F)}=\abs{\Delta_{\X}(M)}.
$$
As a consequence,
$$
d_r(\cdx)=n_1\cdots n_m-\max\left\{\abs{\Delta_\X(N)}, N \in \mathcal{N}_{r,d}\right\}=\min \{ \abs{\nabla_{\X}(N)}, N \in \mathcal{N}_{r,d}\}.
$$
\end{theorem}
\begin{proof}
Let $A_i=\{a_{i0}=0,a_{i1},\dots,a_{i(n_i-1)}\}$, for $1\leq i \leq m$. We consider the bijection
$$
\begin{array}{lccc}
{\rm \varphi}\colon &\X &\rightarrow& \Delta(I(\X)),\\
&(a_{1,i_1},\dots,a_{m,i_m}) & \mapsto& x_1^{i_1}\cdots x_m^{i_m}.
\end{array}
$$
We will prove that $\varphi$ restricts to a biyection from $V_{\X}(M)$ to $\Delta_{\X}(M)$, which gives the result. Let $P=(a_{1,i_1},\dots,a_{m,i_m})\in V_{\X}(M)$. If $x^{\alpha_\ell}\mid \varphi(P)=x_1^{i_1}\cdots x_m^{i_m}$ for some $1\leq \ell \leq r$, then $i_j\geq 1$ for all $j\in \supp(\alpha_\ell)$. This implies that $\supp(\alpha_\ell)\subset \supp(P)$, and $x^{\alpha_\ell}(P)=\prod_{j\in \supp(\alpha_\ell)}a_{j,i_j}\neq 0$, a contradiction with the fact that $P\in V_{\X}(M)$. 

Now let $x^\beta\in \Delta_{\X}(M)$. Let $1\leq i \leq r$. There is $1\leq j \leq m$ such that $x_j\mid x^{\alpha_i}$ but $x_j\nmid x^\beta$. This entails that $\varphi^{-1}(x^\beta)$ has a zero in position $j$, and thus $x^{\alpha_i}(\varphi^{-1}(x^\beta))=0$. Since we arrive at this conclusion for any $1\leq i \leq r$, we obtain $\varphi^{-1}(x^\beta)\in V_{\X}(M)$.
\end{proof}

\begin{remark}
The proof of Theorem \ref{t:fpsharp_homogeneous_d} works for any set of square-free monomials, evaluating over any set $X$ with the property that if $(P_1,\dots, P_m)\in X$, then we also have $(P_1,\dots, P_{i-1},0, P_{i+1},\dots, P_m)\in X$, for any $1\leq i \leq m$. 
\end{remark}

\section{Formulas for the GHWs}
\label{formulas}
While the previous section provides an efficient method to compute the  GHWs of square-free codes using the footprint bound, obtaining explicit formulas for some cases is still interesting, which is what we provide in this section. In what follows, we identify monomials and tuples via the map $x_1^{\alpha_1}\cdots x_m^{\alpha_m}\mapsto (\alpha_1,\dots,\alpha_m)$. Thus, we may say that a tuple divides another tuple, meaning that the corresponding monomial of the former divides the corresponding monomial of the latter. Similarly, we may interchangeably compute the shadow of a set of tuples or a set of monomials. We assume we are using degree lexicographic ordering when considering an ordering on the tuples. This means that, given $\alpha=(\alpha_1,\dots,\alpha_m)$, $\beta=(\beta_1,\dots,\beta_m)$, we have
$$
\begin{aligned}
\alpha >\beta \iff & \abs{\alpha}>\abs{\beta}\text{ or } \abs{\alpha}=\abs{\beta} \text{ and }\alpha_i>\beta_i \text{ for the first index $i$ such that }\alpha_i\neq\beta_i.
\end{aligned}
$$
For the following result, we consider the notation $[m]:=\{1,\dots,m\}$. 

\begin{lemma}\label{l:bound_for_fp}
Let $1\leq d \leq m$ and $1\leq r \leq \binom{m}{d}$ such that $1\leq r \leq m+1-d$. Assume that
\begin{equation}\label{eq:condition_sizes_cartesian}
\frac{n_{d+r-1}n_{d-1}}{n_in_j}\leq \frac{n_{d+r-1}+n_{d-1}-1}{n_i+n_j-1},
\end{equation}
for every $d\leq i<j\leq d+r-2$. Let $M\subset \mathbb{M}_d^S$ with $\abs{M}=r$. Then 
$$
\abs{\nabla_\X(M)}\geq \prod_{i=1}^{d-1}(n_i-1)\prod_{i=d+r}^mn_i\left(\prod_{i=d}^{d+r-1}n_i-1\right).
$$
\end{lemma}
\begin{proof}
We argue by induction on $r$. It is straightforward to check that the result is true for $r=1$. Now we assume it is true for $r-1$. Consider $M=\{a_1,\dots,a_r\}$, with $a_1>a_2>\cdots >a_r$. Then
$$
\abs{\nabla_\X(M)}=\abs{\nabla_\X(a_1,\dots,a_{r-1})}+\abs{\nabla_\X(a_r)\setminus \nabla_\X(a_1,\dots,a_{r-1})}.
$$
By the induction hypothesis, we have
\begin{equation}\label{eq:induction_hypothesis}
\abs{\nabla_\X(a_1,\dots,a_{r-1})}\geq \prod_{i=1}^{d-1}(n_i-1)\prod_{i=d+r-1}^mn_i\left(\prod_{i=d}^{d+r-2}n_i-1\right).
\end{equation}
For each $1\leq i \leq r-1$, we choose $j_i\in \supp(a_i)\setminus \supp(a_r)$, and consider $J=\{j_1,\dots,j_{r-1}\}$ (note that we may have $j_i=j_\ell$ with $i\neq \ell$, so $\abs{J}\leq r-1$). If we consider tuples whose entries corresponding to $J$ are 0, those tuples belong to $B_{\X}\setminus \nabla_\X(a_1,\dots,a_{r-1})$. Then
\begin{equation}\label{eq:bound_nabla}
\abs{\nabla_\X(a_r)\setminus \nabla_\X(a_1,\dots,a_{r-1})} \geq \prod_{i\in \supp(a_r)}(n_i-1)\prod_{i\in [m]\setminus (\supp(a_r) \cup J)}n_i.
\end{equation}
We have that $\abs{[m]\setminus (\supp(a_r) \cup J)}\geq m-d-r+1$. Let $a_r$ be fixed, consider $\ell_1<\cdots <\ell_{r-1}$, the first $r-1$ elements in $[m]\setminus \supp(a_r)$ (we can do this since $r-1\leq m-d$),
and let $u=\max\{i:i\in \supp(a_r)\}$. Then we consider the following elements (recall that we treat monomials and tuples interchangeably)
$$
a'_i=a_r\frac{x_{\ell_i}}{x_{u}}, \; 1\leq i \leq r-1.
$$
We claim that these elements minimize (\ref{eq:bound_nabla}), for this fixed $a_r$. This is equivalent to minimizing
$$
\prod_{i\in [m]\setminus (\supp(a_r) \cup J)}n_i.
$$
Indeed, in this case we can choose $J=\{\ell_1,\dots,\ell_{r-1}\}$. Since $J\cap \supp(a_r)=\emptyset$, this minimizes the number of indices considered in this product. Moreover, this also minimizes the product itself, since $n_1\leq \cdots \leq n_m$, and by construction, we are considering $\ell_1,\cdots,\ell_{r-1}$ the smallest possible indices. Therefore, we have obtained
\begin{equation}\label{eq:bound_nabla_2}
\abs{\nabla_\X(a_r)\setminus \nabla_\X(a_1,\dots,a_{r-1})} \geq \prod_{i\in \supp(a_r)}(n_i-1)\prod_{i\in [m]\setminus (\supp(a_r) \cup \{\ell_1,\dots,\ell_{r-1}\})}n_i.
\end{equation}
Note that, since $a'_1,\dots,a'_{r-1}$ minimize the bound, then the bound obtained for $a'_1,\dots,a'_{r-1}$ also holds for $a_1,\dots,a_{r-1}$. Now we claim that
\begin{equation}\label{eq:bound_difference}
\prod_{i\in \supp(a_r)}(n_i-1)\prod_{i\in [m]\setminus (\supp(a_r) \cup \{\ell_1,\dots,\ell_{r-1}\})}n_i \geq \prod_{i=1}^{d-1}(n_i-1) \prod_{i=d+r}^m n_i (n_{d+r-1}-1).
\end{equation}
The number of factors on both sides is $m-r+1$. In particular, we have $d$ factors of the form $n_i-1$, and $m-d-r+1$ of the form $n_i$. 
Let $t$ be the position, in $a_r$, of the first nonzero entry 
after the first zero entry. This would not make sense only for the 
greatest element of $S_d$ but we are in the case where $r \geq 2$. 
Let $\tilde{a}_r$ be the $m$-tuple whose entries coincide with $a_r$, 
except in the $t$-th position, where now we have a zero, and in the 
$(t-1)$-th position, where now we have a 1. 
Note that $\tilde{a}_r > a_r$.
We write $\tilde{\ell}_1, \ldots, \tilde{\ell}_{r-1}$ to 
denote the smallest possible indices which are not in the support
of $\tilde{a}_r$. We will show that
\begin{equation*} 
\begin{split}
\prod_{i\in \supp(a_r)} & (n_i-1)  \prod_{i\in [m]\setminus (\supp(a_r) \cup \{\ell_1,\dots,\ell_{r-1}\})}n_i \\ & = \left(
\prod_{i\in \supp(a_r)\setminus\{t\}}(n_i-1) \right)  (n_t - 1) \prod_{i\in [m]\setminus (\supp(a_r) \cup \{\ell_1,\dots,\ell_{r-1}\})}  n_i \\
& \geq  \left(
\prod_{i\in \supp(a_r)\setminus\{t\}}(n_i-1) \right)  (n_{t-1} - 1)
\prod_{i\in [m]\setminus (\supp(\tilde{a}_r) \cup \{\tilde{\ell}_1,\dots,\tilde{\ell}_{r-1}\})}  n_i 
\end{split}
\end{equation*}
Let $L = \supp(a_r) \cup  \{\ell_1,\dots,\ell_{r-1}\}$ and
$\tilde{L} = \supp(\tilde{a}_r) \cup \{\tilde{\ell}_1,\dots,\tilde{\ell}_{r-1}\}$.
We know that $t \in L$, $t - 1 \in \tilde{L}$ and
$t-1 \in L$ if and only if $t \in \tilde{L}$. In the case where 
$t$ and $t-1$ both belong to $L$ and $\tilde{L}$, the inequality 
is clear because $L = \tilde{L}$. In case $t-1 \notin L$ then 
we have a factor in 
\[
\prod_{i\in [m]\setminus (\supp(a_r) \cup \{\ell_1,\dots,\ell_{r-1}\})}  n_i 
\]
which is equal to $n_{t - 1}$, and replacing this factor by $n_t$ we
get
\[
\prod_{i\in [m]\setminus (\supp(\tilde{a}_r) \cup \{\tilde{\ell}_1,\dots,\tilde{\ell}_{r-1}\})}  n_i .
\]
Thus the inequality holds because 
$(n_t - 1) n_{t - 1} \geq  (n_{t-1} - 1) n_t$. By iterating this procedure, we move the first indices in $\supp(a_r)$ to $\{1,\dots,d-2\}$, as long as the resulting tuple satisfies $\tilde{a}_r\leq L_r$, where $L_r=x_1\cdots x_{d-1}x_{d+r-1}$ is the $r$th monomial with $\deg(L_r)=d$ in descending lexicographic order. Then $\tilde{a}_r$, seen as a monomial, is of the form $x_1\cdots x_{d-2}x_i x_j$, with $i<j$. If $i=d-1$, then $j\geq d+r-1$ because $a_1>\cdots >a_r$, which implies $L_r\geq a_r$, and therefore $\supp(a_r)$ contains at least two indices larger than $d-2$. In that case, we directly obtain inequality \eqref{eq:bound_difference}. If $i\geq d$ and $j\leq d+r-1$, then, taking into account \eqref{eq:condition_sizes_cartesian},
we have $(n_{d-1}-1)n_in_j(n_{d+r-1}-1)\leq n_{d-1}(n_i-1)(n_j-1)n_{d+r-1}$, which gives \eqref{eq:bound_difference}. Note that, for $j=d+r-1$, the inequality is always true and we do not need to have that case as an extra hypothesis. 
We finish the proof by noting that adding (\ref{eq:induction_hypothesis}) and (\ref{eq:bound_difference}) we get the stated bound. 
\end{proof}

\begin{remark}\label{r:condition}
Condition \eqref{eq:condition_sizes_cartesian} does not apply when $r\leq 2$. This condition implies a moderate growth for the sizes of the sets $A_i$, $d-1\leq i \leq d+r-1$. In particular, if the sizes of those sets are equal, the condition is satisfied. 
\end{remark}

\begin{corollary}\label{c:ghws_homogeneous_d}
Let $1\leq d \leq m$ and $1\leq r \leq m+1-d$. Let $\X=A_1\times \cdots \times A_m$ be a Cartesian set such that $0\in A_i$, $1\leq i \leq m$, and assume that
$$
\frac{n_{d+r-1}n_{d-1}}{n_in_j}\leq \frac{n_{d+r-1}+n_{d-1}-1}{n_i+n_j-1},
$$
for every $d\leq i<j\leq d+r-2$. Then 
$$
d_r(\cdx)=\prod_{i=1}^{d-1}(n_i-1)\prod_{i=d+r}^mn_i\left(\prod_{i=d}^{d+r-1}n_i-1\right).
$$
\end{corollary}
\begin{proof}
By Theorem \ref{t:fpsharp_homogeneous_d} and Corollary \ref{l:bound_for_fp}, it is enough to find $M\subset \mathbb{M}_d^S$ achieving the stated bound. Consider 
$$
S_d(r)=\{x_1\cdots x_{d-1}x_d,x_1\cdots x_{d-1}x_{d+1},\dots, x_1\dots x_{d-1}x_{m-r}\}.
$$
We use the inclusion-exclusion principle to compute $\abs{\nabla_\X(S_d(r))}$. We have that
$$
\abs{\nabla_{\X}(x_1\cdots x_{d-1}x_{d+i-1})}=(n_{d+i-1}-1)\prod_{j=1}^{d-1}(n_j-1) \prod_{j=d, j\neq d+i-1}^m n_j.
$$
Similarly, we obtain that the number multiples of both $x_1\cdots x_{d-1}x_{d+i-1}$ and $x_1\cdots x_{d-1}x_{d+\ell-1}$ in $\Delta(I(\X))$ is
$$
(n_{d+i-1}-1)(n_{d+\ell-1}-1)\prod_{j=1}^{d-1}(n_j-1) \prod_{j=d, j\not \in \{d+i-1,d+\ell-1\}}^m n_j.
$$
In general, using the inclusion-exclusion principle, we get
$$
\begin{aligned}
\abs{\nabla_{\X}(S_d(r))}=\prod_{j=1}^{d-1}(n_j-1)\sum_{A\subset \{d,\dots,d+r-1\},A\neq \emptyset}(-1)^{\abs{A}-1} \prod_{j\in A}(n_j-1) \prod_{j=d, j\not \in A}^m n_j \\
=\prod_{j=1}^{d-1}(n_j-1) \prod_{j=d+r}^m n_j\sum_{A\subset \{d,\dots,d+r-1\},A\neq \emptyset}(-1)^{\abs{A}-1} \prod_{j\in A}(n_j-1) \prod_{j=d, j\not \in A}^{d+r-1} n_j.
\end{aligned}
$$
The fact that 
$$
\sum_{A\subset \{d,\dots,d+r-1\},A\neq \emptyset}(-1)^{\abs{A}-1} \prod_{j\in A}(n_j-1) \prod_{j=d, j\not \in A}^{d+r-1} n_j= \prod_{j=d}^{d+r-1}n_j-1
$$
follows by expanding the product $\prod_{j=d}^{d+r-1}(n_j-(n_j-1))=1$. 

\end{proof}

\begin{corollary}\label{c:ghws_leq_d}
Let $1\leq d \leq m$ and $1\leq r \leq m+1-d$. Assume that
$$
\frac{n_{d+r-1}n_{d-1}}{n_in_j}\leq \frac{n_{d+r-1}+n_{d-1}-1}{n_i+n_j-1},
$$
for every $d\leq i<j\leq d+r-2$. Then
$$
d_r(\cdxx)=\prod_{i=1}^{d-1}(n_i-1)\prod_{i=d+r}^mn_i\left(\prod_{i=d}^{d+r-1}n_i-1\right).
$$
\end{corollary}
\begin{proof}
By Theorem \ref{t:fpsharp_leq_d}, if we prove that, for any $M\subset \mathbb{M}^S_{\leq d}$ with $\abs{M}=r$, we have that
$$
\abs{\nabla_{\X}(M)}\geq \prod_{i=1}^{d-1}(n_i-1)\prod_{i=d+r}^mn_i\left(\prod_{i=d}^{d+r-1}n_i-1\right),
$$
then we obtain the result by considering the same set of monomials as in Corollary \ref{c:ghws_homogeneous_d}. Let $M=\{a_1,\dots,a_r\}$. If $\abs{a_j}<d$, for some $1\leq j \leq r$, then 
$$
\abs{\nabla_\X(M)}\geq \abs{\nabla_\X(x^{a_j})}\geq \prod_{i=1}^{\abs{a_j}}(n_i-1)\prod_{i=\abs{a_j}+1}^m n_i\geq \prod_{i=1}^{d-1}(n_i-1)\prod_{i=d+r}^mn_i\left(\prod_{i=d}^{d+r-1}n_i-1\right).
$$
Indeed, the bound for $\abs{\nabla_\X(x^{a_j})}$ follows from Lemma \ref{l:bound_for_fp}. For the last inequality, factoring out the common terms, we have that the inequality holds if and only if
$$
\prod_{i=\abs{a_j}+1}^{d+r-1}n_i \geq \prod_{i=\abs{a_j}+1}^{d-1}(n_i-1)\left(\prod_{i=d}^{d+r-1}n_i-1\right),
$$
which holds because $\prod_{i=\abs{a_j}+1}^{d-1}n_i> \prod_{i=\abs{a_j}+1}^{d-1}(n_i-1)$ and $\prod_{i=d}^{d+r-1}n_i > \left(\prod_{i=d}^{d+r-1}n_i-1\right)$.
\end{proof}

Similar results to Corollary \ref{c:ghws_homogeneous_d} have been obtained for toric codes in \cite{jaramillo, GHWsToricSquarefree}. In that case, $\X=(\fq^*)^{m}$, and in principle Theorem \ref{t:fpsharp_homogeneous_d} does not apply with equality. However, it still applies as a lower bound, and it turns out to be sharp for some of the cases considered in \cite{jaramillo, GHWsToricSquarefree}, where this is proven by considering appropriate families of polynomials achieving the bound. Our approach of first proving that the footprint bound is sharp and then computing the set of monomials by giving a minimal bound does not generally work for that case, as we show in the next example, where the footprint bound is not sharp.

\begin{example}
Consider $q=3$, $m=5$, $d=2$, $r=3$, and $\X=(\F_q^*)^m$. Using the package from \cite{sanjoseGHWsPackage,githubGHWs}, we see that $d_3(\cdx)=16$. On the other hand, one can check that $\abs{\nabla(S_2(3))}=\abs{\nabla(x_1x_2,x_1x_3,x_1x_4)}=14$. Thus, the footprint bound is not sharp in this case. This fact was observed for the minimum distance when $d>m/2$ in \cite{jaramillo}.
\end{example}

\section{Codes over the projective space}
\label{projective}
In this section, we see how some of the results corresponding to the codes $\cdx$ can be translated to evaluation codes over the projective space. Consider $\PP^m(\fq)$, the rational points of the $m$-dimensional projective space, and fix a set of representatives, which we denote by $P^m$. Then we can define 
$$
\cdp(m):=\langle \ev_{P^m}(x^\alpha): \alpha \in S_d(m+1) \rangle .
$$
A different choice of representatives $P^m$ would lead to a monomially equivalent code. The length of the code is $\abs{P^m}=\frac{q^{m+1}-1}{q-1}$, and the dimension is $\binom{m+1}{d}$. If we define $\A^m_0=\mathbb{A}^m(\fq)\setminus \{0\}$, then we can consider
$$
\cda{m}:=\langle \ev_{\A^m_0}(x^\alpha): \alpha \in S_d(m) \rangle.
$$
Let $\xi\in\fq$ be a primitive element. We denote $\xi \cdot P^m:=\{\xi Q:Q\in P^m\}$. Then we have (for example, see \cite{sanjoseRecursivePRM})
\begin{equation}\label{eq:decomposition_affine_space}
\mathbb{A}^{m+1}(\fq)=P^m\cup \xi \cdot P^m\cup \xi^2\cdot P^m\cup \cdots \cup \xi^{q-2}\cdot P^m\cup \{(0,\dots,0)\}.
\end{equation}
In what follows, we will assume that we have fixed an ordering for the points of $P^m$, which also fixes an ordering for the points of $\xi^i\cdot P^m$, and we will assume we order the points of $\mathbb{A}^{m+1}(\fq)$ as in (\ref{eq:decomposition_affine_space}). The relation between $\cdp(m)$ and $\cda{m}$ is similar to that between homogeneous Reed-Muller codes and projective Reed-Muller codes, e.g., see \cite{bergerAutomorphismsPRM,duursmaHRM}.

\begin{lemma}\label{l:equivalenceAffineProjective}
Let $1\leq d \leq m+1$, and $\xi\in \fq$ a primitive element. Then
$$
\cda{m+1}=(1,\xi^d,\xi^{2d},\dots,\xi^{d(q-2)})\otimes \cdp(m).
$$
As a consequence, $\dim(\cda{m+1})=\dim (\cdp(m))$, and, for every $1\leq r \leq \binom{m+1}{d}$, we have $d_r(\cda{m+1})=(q-1)d_r(\cdp(m))$.

\end{lemma}
\begin{proof}
Let $f\in \fq[x_0,\dots,x_m]_d$ be a homogeneous polynomial of degree $d$. Then, taking (\ref{eq:decomposition_affine_space}) into account, we have
$$
\begin{aligned}
\ev_{\A^m_0}(f)&=(\ev_{P^m}(f),\ev_{\xi \cdot P^m}(f),\ev_{\xi^2 \cdot P^m}(f),\dots,\ev_{\xi^{(q-2)} \cdot P^m}(f))\\
&=(\ev_{P^m}(f),\xi^d\ev_{P^m}(f),\xi^{2d}\ev_{ P^m}(f),\dots,\xi^{d(q-2)}\ev_{P^m}(f))\\
&=(1,\xi^d,\xi^{2d},\dots,\xi^{d(q-2)})\otimes \ev_{P^m}(f).
\end{aligned}
$$
\end{proof}

\begin{remark}
For $d=1$, the codes $\cdp(m)$ are non-degenerate. But for $1<d\leq m+1$, these codes are degenerate, since all the polynomials of $S_d(m+1)$ vanish at any point with $m+2-d$ zero entries, or, equivalently, at every point with Hamming weight lower than or equal to $d-1$. After removing these coordinates, the length of $\cdp(m)$ would be
$$
\frac{q^{m+1}-1}{q-1}-\sum_{i=1}^{d-1} \binom{m+1}{i}.
$$
\end{remark}

From the connection given in Lemma \ref{l:equivalenceAffineProjective}, we can directly obtain the GHWs of the codes $\cda{m}$ and $\cdp(m)$ for some cases.

\begin{corollary}
Let $1\leq d \leq m$ and $1\leq r \leq m+1-d$. We have that 
$$
d_r(\cda{m})=(q-1)^{d-1} q^{m-d-r+1}(q^{r}-1).
$$
As a consequence, for $1\leq d \leq m+1$ and $1\leq r\leq m+2-d$, we have 
$$
d_r(\cdp(m))=(q-1)^{d-2} q^{m-d-r+2}(q^{r}-1).
$$
\end{corollary}
\begin{proof}
From Corollary \ref{c:ghws_homogeneous_d} and Remark \ref{r:condition} we obtain that 
$$
d_r(C_d^{\mathbb{A}^{m+1},S})=(q-1)^{d-1} q^{m-d-r+1}(q^{r}-1).
$$
Now we note that $d_r(C_d^{\mathbb{A}^{m+1}, S})=d_r(\cda{m})$, since all the homogeneous polynomials of degree $d$ vanish at the point $(0,\dots,0)$, implying that the puncturing at the corresponding coordinate preserves the support. The result for $\cdp(m)$ follows from Lemma \ref{l:equivalenceAffineProjective}.

\end{proof}

\section*{Declarations}
\subsection*{Conflict of interest} The authors declare no conflict of interest.


\end{document}